\documentclass[submission,copyright]{eptcs}
\providecommand{\event}{FICS 2015} 

\usepackage{verbatim}
\usepackage{amsmath}
\usepackage{amsthm}
\usepackage{amssymb}
\usepackage{latexsym}
\usepackage{stmaryrd}
\usepackage{breakurl}

\newtheorem{theorem}{Theorem}
\newtheorem{lemma}{Lemma}

\theoremstyle{definition}
\newtheorem{definition}{Definition}

\newtheorem{example}{Example}

\newcommand{\mo}[1]{\llbracket#1\rrbracket}
\newcommand{\mwrt}[2]{\llbracket#1\rrbracket(#2)}
\newcommand{\mwrs}[3]{\llbracket#1\rrbracket_{#3}(#2)}

\newcommand{\aleq}[1][]{\sqsubseteq_{#1}}

\newcommand{\ee}[1][I]{\rhd_{#1}}

\newcommand{\bezem}{\mathcal{M}_\mathsf{Gr(P)}}
\newcommand{\bezemTP}{\mathcal{T}_\mathsf{Gr(P)}}

\title{Equivalence of two Fixed-Point Semantics for Definitional Higher-Order
Logic Programs\thanks{This research was supported by the project ``Handling Uncertainty in Data Intensive Applications'',
co-financed by the European Union (European Social Fund) and Greek national funds,
through the Operational Program ``Education and Lifelong Learning'' of the National
Strategic Reference Framework (NSRF) - Research Program: THALES, Investing in
knowledge society through the European Social Fund.}}
\author{Angelos Charalambidis
\institute{University of Athens \\ Athens, Greece}
\email{a.charalambidis@di.uoa.gr}
\and
Panos Rondogiannis
\institute{University of Athens \\ Athens, Greece}
\email{prondo@di.uoa.gr}
\and
Ioanna Symeonidou
\institute{University of Athens \\ Athens, Greece}
\email{i.symeonidou@di.uoa.gr}
}
\def\titlerunning{Equivalence of two Definitional Higher-Order LP Semantics}
\def\authorrunning{A. Charalambidis, P. Rondogiannis \& I. Symeonidou}

\begin{document}
\maketitle

\begin{abstract}
  Two distinct research approaches have been proposed for assigning a purely
  extensional semantics to higher-order logic programming. The former approach
  uses classical domain-theoretic tools while the latter builds on a fixed-point
  construction defined on a syntactic instantiation of the source program. The
  relationships between these two approaches had not been investigated until now.
  In this paper we demonstrate that for a very broad class of programs, namely
  the class of {\em definitional programs} introduced by W. W. Wadge, the two
  approaches coincide (with respect to ground atoms that involve symbols of the
  program). On the other hand, we argue that if existential higher-order
  variables are allowed to appear in the bodies of program rules, the two approaches
  are in general different. The results of the paper contribute to a better
  understanding of the semantics of higher-order logic programming.
\end{abstract}

\section{Introduction}
Extensional higher-order logic programming has been proposed~\cite{Wa91a,Bezem99,Bezem01,KRW05,CharalambidisHRW13,CharalambidisER14}
as a promising generalization of classical logic programming. The key idea behind this paradigm is
that the predicates defined in a program essentially denote sets and therefore one can use standard
extensional set theory in order to understand their meaning and reason about them. The main difference
between the extensional and the more traditional {\em intensional} approaches to higher-order logic 
programming~\cite{MN2012,CKW93-187} is that the latter approaches have a much richer syntax and 
expressive capabilities but a non-extensional semantics.

Actually, despite the fact that only very few articles have been written regarding extensionality in higher-order logic programming, two main semantic approaches can be identified. The work described in~\cite{Wa91a,KRW05,CharalambidisHRW13,CharalambidisER14}
uses classical domain-theoretic tools in order to capture the meaning of higher-order logic programs.
On the other hand, the work presented in~\cite{Bezem99,Bezem01} builds on a fixed-point
construction defined on a syntactic instantiation of the source program
in order to achieve an extensional semantics. Until now, the relationships between
the above two approaches had not yet been investigated.

In this paper we demonstrate that for a very broad class of programs, namely
the class of {\em definitional programs} introduced by W. W. Wadge~\cite{Wa91a},
the two approaches coincide. Intuitively, this means that for any given definitional
program, the sets of true ground atoms of the program are identical under the two
different semantic approaches. This result is interesting since it suggests that
definitional programs are of fundamental importance for the further study of extensional
higher-order logic programming. On the other hand, we argue that if we try to slightly
extend the source language, 
the two approaches give different
results in general. Overall, the results of the paper contribute to a better
understanding of the semantics of higher-order logic programming and pave the
road for designing a realistic extensional higher-order logic programming language.

The rest of the paper is organized as follows. Section~\ref{section2} briefly introduces
extensional higher-order logic programming and presents in an intuitive way the two existing
approaches for assigning meaning to programs of this paradigm. Section~\ref{section3} contains 
background material, namely the syntax of definitional programs and the formal details behind the
two aforementioned semantic approaches. Section~\ref{section4} demonstrates the
equivalence of the two semantics for definitional programs. Finally, Section~\ref{section5}
concludes the paper with discussion regarding non-definitional programs and
with pointers to future work.

\section{Intuitive Overview of the two Extensional Approaches}\label{section2}
In this section we introduce extensional higher-order logic programming and 
present the two existing approaches for assigning meaning to programs of this paradigm. 
Since these two proposals were initially introduced by W. W. Wadge and M. Bezem
respectively, we will refer to them as {\em Wadge's semantics} and {\em Bezem's semantics}
respectively. The key idea behind both approaches is that in order to achieve
an extensional semantics, one has to consider a fragment of higher-order logic
programming that has a restricted syntax.

\subsection{Extensional Higher-Order Logic Programming}
The main differences between extensional and intensional higher-order logic programming 
can be easily understood through two simple examples (borrowed from~\cite{CharalambidisHRW13}).
Due to space limitations, we avoid a more extensive discussion of this issue; the interested
reader can consult~\cite{CharalambidisHRW13}.
\begin{example}\label{intensional-predicate}
Suppose we have a database of professions, both of their membership
and their status. We might have rules such as:
\[
\begin{array}{l}
\mbox{\tt engineer(tom).}\\
\mbox{\tt engineer(sally).}\\
\mbox{\tt programmer(harry).}
\end{array}
\]
with {\tt engineer} and {\tt programmer} used as predicates. In
intensional higher-order logic programming we could also have rules
in which these are arguments, eg:
\[
\begin{array}{l}
\mbox{\tt profession(engineer).}\\
\mbox{\tt profession(programmer).}
\end{array}
\]
Now suppose {\tt tom} and {\tt sally} are also avid users of
Twitter. We could have rules:
\[
\begin{array}{l}
\mbox{\tt tweeter(tom).}\\
\mbox{\tt tweeter(sally).}
\end{array}
\]
The predicates {\tt tweeter} and {\tt engineer} are
equal as sets (since they are true for the same objects, namely {\tt
tom} and {\tt sally}). If we attempted to understand the above
program from an extensional point of view, then we would have to
accept that {\tt profession(tweeter)} must also hold (since {\tt
tweeter} and {\tt engineer} are indistinguishable as sets). It is
clear that the extensional interpretation in this case is completely
unnatural. The program can however be understood intensionally: the
predicate {\tt profession} is true of the {\em name} {\tt engineer}
(which is different than the name {\tt tweeter}).\qed
\end{example}
On the other hand, there are cases where predicates can be
understood extensionally:
\begin{example}\label{extensional-predicate}
Consider a program that consists only of the following rule:
\[
\begin{array}{l}
\mbox{\tt p(Q):-Q(0),Q(1).}
\end{array}
\]
In an extensional language, predicate {\tt p} above can be
intuitively understood in purely set-theoretic terms: {\tt p} is the
set of all those sets that contain both {\tt 0} and {\tt 1}.

It should be noted that the above program is also a syntactically
acceptable program of the existing intensional logic programming
languages. The difference is that in an extensional language the
above program has a purely set-theoretic semantics.\qed
\end{example}

From the above examples it can be understood that extensional higher-order
logic programming sacrifices some of the rich syntax of intensional higher-order
logic programming in order to achieve semantic clarity.

\subsection{Wadge's Semantics}
The first proposal for an extensional semantics for higher-order logic programming
was given in~\cite{Wa91a} (and later refined and extended in~\cite{KRW05,CharalambidisHRW13,CharalambidisER14}).
The basic idea behind Wadge's approach is that if we consider a properly restricted
higher-order logic programming language, then we can use standard ideas from
denotational semantics in order to assign an extensional meaning to programs.
The basic syntactic assumptions introduced by Wadge in~\cite{Wa91a} are the following:
\begin{itemize}
\item In the head of every rule in a program, each argument of predicate type must be a variable; all such variables must be distinct.
\item The only variables of predicate type that can appear in the body of a
      rule, are variables that appear in its head.
\end{itemize}
Programs that satisfy the above restrictions are named {\em definitional} in~\cite{Wa91a}.
\begin{example}\label{example1}
The program\footnote{For simplicity reasons, the syntax that we use in our example programs is Prolog-like.
The syntax that we adopt in the next section is slightly different and more convenient for the theoretical
developments that follow.}:
\[
\begin{array}{l}
\mbox{\tt p(a).}\\
\mbox{\tt q(b).}\\
\mbox{\tt r(P,Q):-P(a),Q(b).}
\end{array}
\]
is definitional because the arguments of predicate type in the head of the rule for {\tt r}
are distinct variables. Moreover, the only predicate variables that appear in the body of the
same rule, are the variables in its head (namely {\tt P} and {\tt Q}). \qed
\end{example}

\begin{example}\label{example2}
The program:
\[
\begin{array}{l}
\mbox{\tt q(a).}\\
\mbox{\tt r(q).}
\end{array}
\]
is not definitional because the predicate constant {\tt q} appears as an
argument in the second clause. For a similar reason, the program in
Example~\ref{intensional-predicate} is not definitional. The program:
\[
\begin{array}{l}
\mbox{\tt p(Q,Q):-Q(a).}
\end{array}
\]
is also not definitional because the predicate variable {\tt Q} is used twice in the
head of the above rule. Finally, the program:
\[
\begin{array}{l}
\mbox{\tt p(a):-Q(a).}
\end{array}
\]
is not definitional because the predicate variable {\tt Q} that appears in the
body of the above rule, does not appear in the head of the rule.\qed
\end{example}
As it is argued in~\cite{Wa91a}, if a program satisfies the above two syntactic restrictions,
then it has a {\em unique minimum model} (this notion will be precisely defined in
Section~\ref{section3}). Consider again the program of Example~\ref{example1}. In the minimum
model of this program, the meaning of predicate {\tt p} is the relation $\{{\tt a}\}$ and
the meaning of predicate {\tt q} is the relation $\{{\tt b}\}$. On the other hand, the meaning
of predicate {\tt r} in the minimum model is a relation that contains the pairs $(\{{\tt a}\},\{{\tt b}\})$,
$(\{{\tt a,b}\},\{{\tt b}\})$, $(\{{\tt a}\},\{{\tt a,b}\})$ and $(\{{\tt a,b}\},\{{\tt a,b}\})$.
As remarked by W. W. Wadge (and formally demonstrated in~\cite{KRW05,CharalambidisHRW13}), the minimum
model of every definitional program is monotonic and continuous\footnote{The notion of continuity
will not play any role in the remaining part of this paper.}. Intuitively, monotonicity means
that if in the minimum model the meaning of a predicate is true of a relation, then it is also
true of every superset of this relation. For example, we see that since the meaning of {\tt r}
is true of  $(\{{\tt a}\},\{{\tt b}\})$, then it is also true of $(\{{\tt a,b}\},\{{\tt b}\})$
(because $\{{\tt a,b}\}$ is a superset of $\{{\tt a}\}$).

The minimum model of a given definitional program can be constructed as the least fixed-point
of an operator that is associated with the program, called the {\em immediate consequence operator}
of the program. As it is demonstrated in~\cite{Wa91a,KRW05}, the immediate consequence operator
is monotonic, and this guarantees the existence of the least fixed-point which is constructed
by a bottom-up iterative procedure (more formal details will be given in the next section).
\begin{example}\label{example3}
Consider the definitional program:
\[
\begin{array}{l}
\mbox{\tt q(a).}\\
\mbox{\tt q(b).}\\
\mbox{\tt p(Q):-Q(a).}\\
\mbox{\tt id(R)(X):-R(X).}
\end{array}
\]
In the minimum model of the above program, the meaning of {\tt q} is the relation $\{{\tt a},{\tt b}\}$.
The meaning of {\tt p} is the set of all relations that contain (at least) {\tt a}; more formally,
it is the relation $\{r \mid {\tt a} \in r\}$. The meaning of {\tt id} is the set of all pairs $(r,d)$
such that $d$ belongs to $r$; more formally, it is the relation $\{(r,d)\mid d \in r\}$.\qed
\end{example}
Notice that in the construction of the minimum model, all predicates are initially assigned the
empty relation. The rules of the program are then used in order to improve the meaning assigned to
each predicate symbol. More specifically, at each step of the fixed-point computation, the meaning
of each predicate symbol either stabilizes or it becomes richer than the previous step.
\begin{example}\label{example4}
Consider again the definitional program of the previous example. In the iterative construction
of the minimum model, all predicates are initially assigned the empty relation (of the corresponding
type). After the first step of the construction, the meaning assigned to predicate {\tt q} is
the relation $\{{\tt a},{\tt b}\}$ due to the first two facts of the program. At this same step,
the meaning of {\tt p} becomes the relation $\{r \mid {\tt a} \in r\}$. Also, the meaning of
{\tt id} becomes equal to the relation $\{(r,d)\mid d \in r\}$. Additional iterations will not
alter the relations we have obtained at the first step; in other words, we have reached the
fixed-point of the bottom-up computation.\qed
\end{example}
In the above example, we obtained the meaning of the program in just one step. If the source
program contained recursive definitions, convergence to the least fixed-point would in general
require more steps.

\subsection{Bezem's Semantics}
In~\cite{Bezem99,Bezem01}, M. Bezem proposed an alternative extensional semantics for higher-order
logic programs. Again, the syntax of the source language has to be appropriately restricted.
Actually, the class of programs adopted in~\cite{Bezem99,Bezem01} is a proper superset of the
class of definitional programs. In particular, Bezem proposes the class of {\em hoapata programs} which
extend definitional programs:
\begin{itemize}
\item A predicate variable that appears in the body of a rule, need not
      necessarily appear in the head of that rule.

\item The head of a rule can be an atom that starts with a predicate variable.
\end{itemize}
\begin{example}
All definitional programs of the previous subsection are also hoapata.
The following non-definitional program of Example~\ref{example2} is hoapata:
\[
\begin{array}{l}
\mbox{\tt p(a):-Q(a).}
\end{array}
\]
Intuitively, the above program states that {\tt p} is true of {\tt a}
if there exists a predicate that is defined in the program
that is true of {\tt a}. We will use this program
in our discussion at the end of the paper.

The following program is also hoapata (but not definitional):
\[
\begin{array}{l}
\mbox{\tt P(a,b).}
\end{array}
\]
Intuitively, the above program states that every binary relation is
true of the pair $({\tt a},{\tt b})$.\qed
\end{example}

Given a hoapata program, the starting idea behind Bezem's approach is to take its
``ground instantiation'' in which we replace variables with well-typed terms of
the Herbrand Universe of the program (ie., terms that can be created using only
predicate and individual constants that appear in the program).  For example, given
the program:
\[
\begin{array}{l}
\mbox{\tt q(a).}\\
\mbox{\tt q(b).}\\
\mbox{\tt p(Q):-Q(a).}\\
\mbox{\tt id(R)(X):-R(X).}
\end{array}
\]
the ground instantiation is the following infinite ``program'':
\[
\begin{array}{l}
\mbox{\tt q(a).}\\
\mbox{\tt q(b).}\\
\mbox{\tt p(q):-q(a).}\\
\mbox{\tt id(q)(a):-q(a).}\\
\mbox{\tt p(id(q)):-id(q)(a).}\\
\mbox{\tt id(id(q))(a):-id(q)(a).}\\
\mbox{\tt p(id(id(q))):-id(id(q))(a).}\\
\hspace{2cm} \cdots
\end{array}
\]
One can now treat the new program as an infinite propositional one (ie., each ground atom can be seen
as a propositional one). This implies that we can use the standard least fixed-point construction
of classical logic programming (see for example~\cite{lloyd}) in order to compute the set of atoms
that should be taken as ``true''.  In our example, the least fixed-point will
contain atoms such as {\tt q(a)}, {\tt q(b)}, {\tt p(q)}, {\tt id(q)(a)}, {\tt p(id(q))},
and so on.

A main contribution of Bezem's work was that he established that the least fixed-point
of the ground instantiation of every hoapata program is {\em extensional}. This notion
can intuitively be explained as follows. It is obvious in the
above example that the relations {\tt q} and {\tt id(q)} are equal (they are both true of only the
constant {\tt a}, and therefore they both correspond to the relation $\{{\tt a}\}$).
Therefore, we would expect that (for example) if {\tt p(q)} is true then {\tt p(id(q))}
is also true because {\tt q} and {\tt id(q)} should be considered as interchangeable.
This property of ``interchangeability'' is formally defined in~\cite{Bezem99,Bezem01} and it is
demonstrated that it holds in the least fixed-point of the immediate consequence operator
of the ground instance of every hoapata program.

\subsection{The Differences Between the two Approaches}
It is not hard to see that the two semantic approaches outlined in the previous
subsections, have some important differences. First, they operate on different
source languages. Therefore, in order to compare them we have to restrict Bezem's
approach to the class of definitional programs\footnote{Actually, we could alternatively
extend Wadge's approach to a broader class of programs. Such an extension has already
been performed in~\cite{CharalambidisHRW13}, and we will discuss its repercussions in the
concluding section.}.

The main difference however between the two approaches is the way that the least
fixed-point of the immediate consequence operator is constructed in each case.
In Wadge's semantics the construction starts by initially assigning to every predicate
constant the empty relation; these relations are then improved at each step until
they converge to their final meaning. In other words, Wadge's semantics {\em manipulates
relations}. On the other hand, Bezem's semantics works with the ground instantiation
of the source program and, at first sight, it appears to have a more syntactic flavor.
In our running example, Wadge's approach converges in a single step while Bezem's
approach takes an infinite number of steps in order to converge. However, one can easily verify that
the ground atoms that belong to the least fixed-point under Bezem's semantics, are also
true in the minimum model under Wadge's semantics. This poses the question whether
under both approaches, the sets of ground atoms that are true, are identical.
This is the question that we answer positively in the rest of the paper.

\section{Definitional Programs and their Semantics}
\label{sec:lang}
\label{section3}
In this section we define the source language ${\cal H}$ of definitional
higher-order logic programs. Moreover, we present in a formal way the two
different extensional semantics that have been proposed for such programs,
namely Wadge's and Bezem's semantics respectively.

\subsection{Syntax}
The language ${\cal H}$ is based on a simple
type system that supports two base types: $o$, the boolean domain,
and $\iota$, the domain of individuals (data objects). The composite
types are partitioned into three classes: functional (assigned to
function symbols), predicate (assigned to predicate symbols) and
argument (assigned to parameters of predicates).

\begin{definition}
A type can either be functional, argument, or predicate, denoted by $\sigma$, $\rho$
and $\pi$ respectively and defined as:
\begin{align*}
\sigma & := \iota \mid (\iota \rightarrow \sigma) \\
\pi & := o \mid (\rho \rightarrow \pi) \\
\rho & := \iota \mid \pi
\end{align*}
\end{definition}

We will use $\tau$ to denote an arbitrary type (either functional, argument or predicate one).
As usual, the binary operator $\rightarrow$ is right-associative. A
functional type that is different than $\iota$ will often be written
in the form $\iota^n \rightarrow \iota$, $n\geq 1$. Moreover, it can
be easily seen that every predicate type $\pi$ can be written in the
form $\rho_1 \rightarrow \cdots \rightarrow \rho_n \rightarrow o$,
$n\geq 0$ (for $n=0$ we assume that $\pi=o$).

We proceed by defining the syntax of ${\cal H}$:
\begin{definition}
The alphabet of the higher-order language ${\cal H}$ consists of the following:
\begin{enumerate}
  \item Predicate variables of every predicate type $\pi$ (denoted by capital letters such as
      $\mathsf{P,Q,R,\ldots}$).
  \item Individual variables of type $\iota$ (denoted by capital letters such as
      $\mathsf{X,Y,Z,\ldots}$).
  \item Predicate constants of every predicate type $\pi$ (denoted by lowercase letters such as
      $\mathsf{p,q,r,\ldots}$).
  \item Individual constants of type $\iota$ (denoted by lowercase
      letters such as $\mathsf{a,b,c,\ldots}$).
  \item Function symbols of every functional type $\sigma \neq \iota$ (denoted by lowercase letters such as $\mathsf{f,g,h,\ldots}$).
  \item The logical conjunction constant $\wedge$, the inverse implication constant $\leftarrow$, the left and right parentheses,
        and the equality constant $\approx$ for comparing terms of type $\iota$.
\end{enumerate}
\end{definition}
The set consisting of the predicate variables and the individual variables of ${\cal H}$
will be called the set of {\em argument variables} of ${\cal H}$. Argument variables will
be usually denoted by $\mathsf{V}$ and its subscripted versions.

\begin{definition}
The set of {\em terms} of the higher-order language ${\cal H}$ is defined as follows:
\begin{itemize}
  \item Every predicate variable (respectively predicate constant) of type $\pi$ is a
        term of type $\pi$; every individual variable (respectively individual constant)
        of type $\iota$ is a term of type $\iota$;
  \item if $\mathsf{f}$ is an $n$-ary function symbol and $\mathsf{E}_1, \ldots, \mathsf{E}_n$
        are terms of type $\iota$ then $(\mathsf{f}\ \mathsf{E}_1\cdots\mathsf{E}_n)$ is
        a term of type $\iota$;
  \item if $\mathsf{E}_1$ is a term of type $\rho \rightarrow \pi$ and
        $\mathsf{E}_2$ a term of type $\rho$ then $(\mathsf{E}_1\ \mathsf{E}_2)$ is a term of type $\pi$.
\end{itemize}
\end{definition}
\begin{definition}
The set of {\em expressions} of the higher-order language ${\cal H}$ is defined as follows:
\begin{itemize}
\item A term of type $\rho$ is an expression of type $\rho$;
\item if $\mathsf{E}_1$ and $\mathsf{E}_2$ are terms of type $\iota$, then $(\mathsf{E}_1\approx \mathsf{E}_2)$ is an expression of type $o$.
\end{itemize}
\end{definition}
We write $vars(\mathsf{E})$ to denote the set of all the variables in $\mathsf{E}$.
Expressions (respectively terms) that have no variables will often be referred to as {\em ground expressions} (respectively {\em ground terms}). Expressions of type $o$ will often be referred to as {\em atoms}. We will omit parentheses
when no confusion arises. To denote that an expression $\mathsf{E}$ has type $\rho$ we will often write $\mathsf{E}:\rho$.

\begin{definition}
A {\em clause} is a formula $\mathsf{p}\ \mathsf{V}_1 \cdots \mathsf{V}_n \leftarrow \mathsf{E}_1 \wedge\cdots \wedge \mathsf{E}_m$,
where $\mathsf{p}$ is a predicate constant, $\mathsf{p}\ \mathsf{V}_1\cdots\mathsf{V}_n$
is a term of type $o$ and $\mathsf{E}_1,\ldots,\mathsf{E}_m$ are expressions
of type $o$. The term $\mathsf{p}\ \mathsf{V}_1 \cdots \mathsf{V}_n$ is called the {\em head} of the
clause, the variables $\mathsf{V}_1, \ldots, \mathsf{V}_n$ are the {\em formal parameters} of the
clause and the conjunction $\mathsf{E}_1 \wedge\cdots \wedge \mathsf{E}_m$ is its {\em body}.
A {\em definitional clause} is a clause that additionally satisfies the following
two restrictions:
\begin{enumerate}
\item All the formal parameters are distinct variables (ie., for all $i,j$ such that $1\leq i,j\leq n$,
      $\mathsf{V}_i \neq \mathsf{V}_j$).

\item The only variables that can appear in the body of the clause are its formal parameters
      and possibly some additional individual variables (namely variables of type $\iota$).
\end{enumerate}
A {\em program} $\mathsf{P}$ is a set of definitional program clauses.
\end{definition}
In the rest of the paper, when we refer to ``clauses'' we will mean definitional ones.
For simplicity, we will follow the usual logic programming convention and we will write
$\mathsf{p}\ \mathsf{V}_1 \cdots \mathsf{V}_n \leftarrow \mathsf{E}_1,\ldots,\mathsf{E}_m$
instead of $\mathsf{p}\ \mathsf{V}_1 \cdots \mathsf{V}_n \leftarrow \mathsf{E}_1 \wedge\cdots \wedge \mathsf{E}_m$.

Our syntax differs slightly from the Prolog-like syntax that we have used in Section~\ref{section2}.
However, one can easily verify that we can transform every program from the former syntax to the latter.

%

\begin{definition}
For a program $\mathsf{P}$, we define the Herbrand universe for every argument type $\rho$, denoted by
$U_{\mathsf{P},\rho}$ to be the set of all ground terms of type $\rho$, that can be formed out of the individual constants, function symbols and predicate constants in the program.
\end{definition}
In the following, we will often talk about the ``ground instantiation of a program''. This
notion is formally defined below.
\begin{definition}
A {\em ground substitution} $\theta$ is a finite set of the form $\{ \mathsf{V}_1/\mathsf{E}_1, \ldots, \mathsf{V}_n/\mathsf{E}_n\}$
where the $\mathsf{V}_i$'s are different argument variables and each $\mathsf{E}_i$
is a ground term having the same type as $\mathsf{V}_i$. We write
$dom(\theta) = \{ \mathsf{V}_1, \ldots, \mathsf{V}_n\}$ to denote the domain of $\theta$.
\end{definition}
We can now define the application of a substitution to an expression.
\begin{definition}
Let $\theta$ be a substitution and $\mathsf{E}$ be an expression. Then, $\mathsf{E}\theta$
is an expression obtained from $\mathsf{E}$ as follows:
\begin{itemize}
  \item $\mathsf{E}\theta = \mathsf{E}$ if $\mathsf{E}$ is a predicate or individual constant;
  \item $\mathsf{V}\theta = \theta(\mathsf{V})$ if $\mathsf{V} \in dom(\theta)$; otherwise, $\mathsf{V}\theta = \mathsf{V}$;
  \item $(\mathsf{f}\ \mathsf{E}_1\cdots\mathsf{E}_n)\theta = (\mathsf{f}\ \mathsf{E}_1\theta\cdots\mathsf{E}_n\theta)$;
  \item $(\mathsf{E}_1\ \mathsf{E}_2)\theta = (\mathsf{E}_1\theta\ \mathsf{E}_2\theta)$;
  \item $(\mathsf{E}_1\approx \mathsf{E}_2)\theta = (\mathsf{E}_1\theta\approx \mathsf{E}_2\theta)$.
\end{itemize}
\end{definition}
\begin{definition}
Let $\mathsf{E}$ be an expression and $\theta$ be a ground substitution such that
$vars(\mathsf{E}) \subseteq dom(\theta)$. Then, the ground expression $\mathsf{E}\theta$
is called a {\em ground instantiation} of $\mathsf{E}$. A {\em ground instantiation of a clause}
$\mathsf{p}\ \mathsf{V}_1 \cdots \mathsf{V}_n \leftarrow \mathsf{E}_1,\ldots,\mathsf{E}_m$
with respect to a ground substitution $\theta$ is the formula
$(\mathsf{p}\ \mathsf{V}_1 \cdots \mathsf{V}_n)\theta \leftarrow \mathsf{E}_1\theta,\ldots,\mathsf{E}_m\theta$.
The {\em ground instantiation of a program} $\mathsf{P}$ is the (possibly infinite)
set that contains all the ground instantiations of the clauses of $\mathsf{P}$
with respect to all possible ground substitutions.
\end{definition}

\subsection{Wadge's Semantics}
\label{sec:hosem}
The key idea behind Wadge's semantics is (intuitively) to assign to program
predicates monotonic relations. In the following, given posets $A$ and $B$,
we write $[A \stackrel{m}{\rightarrow} B]$ to denote the set of all monotonic
relations from $A$ to $B$.

Before specifying the semantics of expressions of ${\cal H}$ we need to
provide the set-theoretic meaning of the types of expressions of
${\cal H}$ with respect to an underlying domain. It is customary in
logic programming to take the underlying domain to be the Herbrand
universe $U_{\mathsf{P},\iota}$. In the following definition we define
simultaneously and recursively two things: the semantics $\mo{\tau}$
of a type $\tau$ and a corresponding partial order $\aleq[\tau]$
on the elements of $\mo{\tau}$. We adopt the usual ordering of the truth values
$\mathit{false}$ and $\mathit{true}$, i.e. $\mathit{false} \leq \mathit{false}$,
$\mathit{true}\leq \mathit{true}$ and $\mathit{false} \leq \mathit{true}$.
\begin{definition}
Let $\mathsf{P}$ be a program. Then,
\begin{itemize}
  \item $\mo{\iota} = U_{\mathsf{P},\iota}$ and $\aleq[\iota]$ is the trivial partial order
        that relates every element to itself;
  \item $\mo{\iota^n\rightarrow \iota} =  U_{\mathsf{P},\iota}^n \rightarrow U_{\mathsf{P},\iota}$. A partial order for this case is not needed;
  \item $\mo{o} = \{ \mathit{false}, \mathit{true} \}$ and $\aleq[o]$ is the partial order $\leq$ on truth values;
  \item $\mo{\rho \rightarrow \pi} = [ \mo{\rho} \stackrel{m}{\rightarrow}  \mo{\pi} ]$ and $\aleq[\rho \rightarrow \pi]$
  is the partial order defined as follows: for all $f, g \in \mo{\rho \rightarrow \pi}$,
  $f \aleq[\rho \rightarrow \pi] g$ iff $f(d) \aleq[\pi] g(d)$ for all $d \in \mo{\rho}$.
\end{itemize}
\end{definition}
We now proceed to define Herbrand interpretations and states.
\begin{definition}
A Herbrand interpretation $I$ of a program $\mathsf{P}$ is an interpretation such that:
\begin{enumerate}
  \item for every individual constant $\mathsf{c}$ that appears in $\mathsf{P}$, $I(\mathsf{c}) = \mathsf{c}$;
  \item for every predicate constant $\mathsf{p} : \pi$ that appears in $\mathsf{P}$, $I(\mathsf{p}) \in \mo{\pi}$;
  \item for every $n$-ary function symbol $\mathsf{f}$ that appears in $\mathsf{P}$ and for all
  $\mathsf{t}_1, \ldots \mathsf{t}_n \in U_{\mathsf{P}, \iota}$, $I(\mathsf{f})\ \mathsf{t}_1\ \cdots \mathsf{t}_n = \mathsf{f}\ \mathsf{t}_1\ \cdots\mathsf{t}_n$.
\end{enumerate}
\end{definition}

\begin{definition}
  A Herbrand state $s$ of a program $\mathsf{P}$ is a function that assigns to each argument variable $\mathsf{V}$
  of type $\rho$, an element $s(\mathsf{V}) \in \mo{\rho}$.
\end{definition}
In the following, $s[\mathsf{V}/d]$ is used to denote a state that is identical to $s$ the
only difference being that the new state assigns to $V$ the value $d$.

\begin{definition}
  Let $\mathsf{P}$ be a program, $I$ be a Herbrand interpretation of $\mathsf{P}$ and
  $s$ be a Herbrand state. Then, the semantics of the expressions of $\mathsf{P}$ is
  defined as follows:
\begin{enumerate}
  \item $\mwrs{\mathsf{V}}{I}{s} = s(\mathsf{V})$ if $\mathsf{V}$ is a variable;
  \item $\mwrs{\mathsf{c}}{I}{s} = I(\mathsf{c})$ if $\mathsf{c}$ is an individual constant;
  \item $\mwrs{\mathsf{p}}{I}{s} = I(\mathsf{p})$ if $\mathsf{p}$ is a predicate constant;
  \item $\mwrs{(\mathsf{f}\ \mathsf{E}_1\cdots \mathsf{E}_n)}{I}{s} = I(\mathsf{f})\ \mwrs{\mathsf{E}_1}{I}{s}\cdots \mwrs{\mathsf{E}_n}{I}{s}$;
  \item $\mwrs{(\mathsf{E}_1\ \mathsf{E}_2)}{I}{s} = \mwrs{\mathsf{E}_1}{I}{s}\ \mwrs{\mathsf{E}_2}{I}{s}$;
  \item $\mwrs{(\mathsf{E}_1\approx \mathsf{E}_2)}{I}{s} = true$ if $\mwrs{\mathsf{E}_1}{I}{s}= \mwrs{\mathsf{E}_2}{I}{s}$ and $\mathit{false}$ otherwise.
\end{enumerate}
\end{definition}
For ground expressions $\mathsf{E}$ we will often write $\mwrs{\mathsf{E}}{I}{}$ instead
of $\mwrs{\mathsf{E}}{I}{s}$ since the meaning of $\mathsf{E}$ is independent of $s$.

It is straightforward to confirm that the above definition assigns to every expression an element of the corresponding semantic domain, as stated in the following lemma:
\begin{lemma}
Let $\mathsf{P}$ be a program and let $\mathsf{E} : \rho$ be an expression. Also, let $I$ be a Herbrand interpretation and $s$ be a Herbrand state. Then $\mwrs{\mathsf{E}}{I}{s}\in\mo{\rho}$.
\end{lemma}

\begin{definition}
Let $\mathsf{P}$ be a program and $M$ be a Herbrand interpretation of $\mathsf{P}$.
Then, $M$ is a Herbrand model of $\mathsf{P}$ iff for every clause
$\mathsf{p}\ \mathsf{V}_1\cdots\mathsf{V}_n \leftarrow \mathsf{E}_1, \ldots \mathsf{E}_m$ in $\mathsf{P}$
and for every Herbrand state $s$, if for all $i \in \{1,\ldots,m\}$, $\mwrs{\mathsf{E}_i}{M}{s} = \mathit{true}$ then
$\mwrs{\mathsf{p}\ \mathsf{V}_1\cdots\mathsf{V}_n}{M}{s} = \mathit{true}$.
\end{definition}
In the following we denote the set of Herbrand interpretations of a program $\mathsf{P}$
with ${\cal I}_\mathsf{P}$. We define a partial order on ${\cal I}_\mathsf{P}$ as
follows: for all $I, J \in {\cal I}_\mathsf{P}$, $I \aleq[{\cal I}_\mathsf{P}] J$
iff for every predicate $\mathsf{p} : \pi$ that appears in $\mathsf{P}$, $I(\mathsf{p}) \aleq[\pi] J(\mathsf{p})$.
Similarly, we denote the set of Herbrand states with ${\cal S}_\mathsf{P}$ and we define
a partial order as follows: for all $s_1, s_2 \in {\cal S}_\mathsf{P}$, $s_1 \aleq[{\cal S}_\mathsf{P}] s_2$ iff
for all variables $\mathsf{V} : \rho$, $s_1(\mathsf{V}) \aleq[\rho] s_2(\mathsf{V})$.
The following lemmata are straightforward to establish:
\begin{lemma}
Let $\mathsf{P}$ be a program. Then, $({\cal I}_\mathsf{P},\aleq[{\cal I}_\mathsf{P}])$ is a complete lattice.
\end{lemma}

\begin{lemma} \label{interp-state-monotonicity}
Let $\mathsf{P}$ be a program and let $\mathsf{E} : \rho$ be an expression.
Let $I, J$ be Herbrand interpretations and $s, s'$ be Herbrand states.
Then,
\begin{enumerate}
  \item If $I \aleq[{\cal I}_\mathsf{P}] J$  then $\mwrs{\mathsf{E}}{I}{s} \aleq[\rho] \mwrs{\mathsf{E}}{J}{s}$.
  \item If $s \aleq[{\cal S}_\mathsf{P}] s'$ then $\mwrs{\mathsf{E}}{I}{s} \aleq[\rho] \mwrs{\mathsf{E}}{I}{s'}$.
\end{enumerate}
\end{lemma}

We can now define the {\em immediate consequence operator} for
${\cal H}$ programs, which generalizes the corresponding operator for classical (first-order)
programs~\cite{lloyd}.
\begin{definition}
  Let $\mathsf{P}$ be a program. The mapping $T_\mathsf{P} : {\cal I}_\mathsf{P} \rightarrow {\cal I}_\mathsf{P}$
  is called the {\em immediate consequence operator for $\mathsf{P}$} and is defined for every predicate
  $\mathsf{p} : \rho_1 \rightarrow \cdots \rightarrow \rho_n \rightarrow o$ and $d_i \in \mo{\rho_i}$ as
\[T_\mathsf{P}(I)(\mathsf{p})\ d_1\cdots d_n =
\begin{cases}
  \mathit{true}  & \mbox{there exists a clause $\mathsf{p}\ \mathsf{V}_1\cdots\mathsf{V}_n \leftarrow \mathsf{E}_1, \ldots \mathsf{E}_m$ such that}\\
                 & \mbox{for every state $s$,  $\mwrs{\mathsf{E}_i}{I}{s[\mathsf{V}_1/d_1, \ldots, \mathsf{V}_n/d_n]} = \mathit{true}$ for all $i \in \{1, \ldots, m\}$} \\
  \mathit{false} & \mbox{otherwise.}
\end{cases}
\]
\end{definition}

%
%

It is not hard to see that $T_\mathsf{P}$ is a monotonic function, and this leads to
the following theorem~\cite{Wa91a,KRW05}:
\begin{theorem}
Let $\mathsf{P}$ be a program.
Then $M_{\mathsf{P}} =\mathit{lfp}(T_\mathsf{P})$
is the minimum, with respect to $\aleq[{\cal I}_\mathsf{P}]$, Herbrand model of $\mathsf{P}$.
\end{theorem}


\subsection{Bezem's Semantics}
\label{sec:bezemsem}
In contrast to Wadge's semantics which proceeds by constructing the meaning of predicates
as relations, Bezem's approach takes a (seemingly) more syntax-oriented approach. In particular,
Bezem's approach builds on the ground instantiation of the source program in order to
retrieve the meaning of the program. In our definitions below, we follow relatively closely
the exposition given in~\cite{Bezem99,Bezem01,Bezem2002}.
\begin{definition}
Let $\mathsf{P}$ be a program and let $\mathsf{Gr(P)}$ be its ground instantiation.
An interpretation $I$ for $\mathsf{Gr(P)}$ is defined as a subset of $U_{\mathsf{P},o}$
by the usual convention that, for any $\mathsf{A}\in U_{\mathsf{P},o}$, $I(\mathsf{A})=true$
iff $\mathsf{A} \in I$. We also extend the interpretation $I$ for every $(\mathsf{E}_1 \approx \mathsf{E}_2)$
atom as follows: $I(\mathsf{E_1}\approx\mathsf{E_2}) = \mathit{true}$ if $\mathsf{E_1} = \mathsf{E_2}$ and $\mathit{false}$ otherwise.
\end{definition}

Observe that the meaning of $(\mathsf{E}_1 \approx \mathsf{E}_2)$ is fixed and independent of the interpretation.

\begin{definition}
We define the immediate consequence operator, $\bezemTP$, of $\mathsf{P}$ as follows:
\begin{align*}
\bezemTP(I)(\mathsf{A}) = \begin{cases}
 \mathit{true} & \mbox{if there exists a clause } \mathsf{A} \leftarrow \mathsf{E}_1, \ldots, \mathsf{E}_m \mbox{ in } \mathsf{Gr(P)} \\
 & \mbox{ such that } I(\mathsf{E}_i) =true \mbox{ for all } i \in \{1, \ldots, m\} \\
 \mathit{false} & \mbox{otherwise}.
\end{cases}
\end{align*}
\end{definition}

As it is well established in bibliography (for example~\cite{lloyd}), the least fixed-point of the immediate consequence operator of a propositional program exists and is the minimum, with respect to
set inclusion and equivalently $\leq$, model of $\mathsf{Gr(P)}$. This fixed-point, which we will henceforth denote by $\bezem$, is shown in~\cite{Bezem99,Bezem01} to be directly related to a notion of a model capable of capturing the perceived semantics of the higher-order program $\mathsf{P}$. In particular, this model by definition assigns to all ground atoms the same truth values as $\bezem$. It is therefore justified that we restrict our attention to $\bezem$, instead of the aforementioned higher-order model, in our attempt to prove the equivalence of Bezem's semantics and Wadge's semantics.


The following definition and subsequent theorem obtained in~\cite{Bezem2002},
identify a property of $\bezem$ that we will need in the next section.
%
\begin{definition}
Let $\mathsf{P}$ be a program and let $\bezem$ be the $\leq$-minimum model of $\mathsf{Gr(P)}$.
For every argument type $\rho$ we define a corresponding partial order as follows: for type $\iota$, we define
$\preceq_\iota$ as syntactical equality, i.e. $\mathsf{E} \preceq_\iota \mathsf{E}$ for all $\mathsf{E} \in U_{\mathsf{P},\iota}$.
For type $o$, $\mathsf{E} \preceq_o \mathsf{E'}$ iff $\bezem(\mathsf{E}) \leq \bezem(\mathsf{E'})$. For a predicate type of
the form $\rho \rightarrow \pi$, $\mathsf{E} \preceq_{\rho \rightarrow \pi} \mathsf{E'}$ iff
$\mathsf{E} \, \mathsf{D} \preceq_{\pi} \mathsf{E'} \, \mathsf{D}$ for all $\mathsf{D} \in U_{\mathsf{P},\rho}$.
\end{definition}

%
%
\begin{theorem}[$\preceq$-Monotonicity Property]\label{th-preceq-monotonicity}~\cite{Bezem2002}
Let $\mathsf{P}$ be a program and $\bezem$ be the $\leq$-minimum model of $\mathsf{Gr(P)}$. Then for all $\mathsf{E} \in U_{\mathsf{P},\rho\rightarrow\pi}$ and all $\mathsf{D}, \mathsf{D}' \in U_{\mathsf{P},\rho}$ such that $\mathsf{D} \preceq_\rho \mathsf{D}'$, it holds $\mathsf{E} \, \mathsf{D} \preceq_{\pi} \mathsf{E} \, \mathsf{D}'$.
\end{theorem}

\section{Equivalence of the two Semantics}
\label{sec:positive}
\label{section4}
{In this section we demonstrate that the two semantics presented in the previous
section, are equivalent for definitional programs. To help us transcend the
differences between these approaches, we introduce two key notions, namely that
of the \emph{ground restriction} of a higher-order interpretation and its
complementary notion of the \emph{semantic extension} of ground expressions.
But first we present the following \emph{Substitution Lemma}, which will be useful in the proofs of later results.}
\begin{lemma}[Substitution Lemma]
\label{substitution_lemma}
Let $\mathsf{P}$ be a program and $I$ be a Herbrand interpretation of $\mathsf{P}$. Also
let $\mathsf{E}$ be an expression and $\theta$ be a ground
substitution with $vars(\mathsf{E}) \subseteq dom(\theta)$. If $s$ is a Herbrand state such that, for all
$\mathsf{V} \in vars(\mathsf{E})$, $s(\mathsf{V}) = \mwrt{\theta(\mathsf{V})}{I}$, then
$\mwrs{\mathsf{E}}{I}{s} = \mwrt{\mathsf{E}\theta}{I}$.
\end{lemma}
\begin{proof}
By a structural induction on $\mathsf{E}$. For the basis case, if
$\mathsf{E}=\mathsf{p}$ or $\mathsf{E}=\mathsf{c}$ then
the statement reduces to an identity and if $\mathsf{E}=\mathsf{V}$
then it holds by assumption. For the induction step, we first examine the case that $\mathsf{E}= (\mathsf{f}\ \mathsf{E}_1\ \cdots\ \mathsf{E}_n)$;
then $\mwrs{\mathsf{E}}{I}{s} = I(\mathsf{f})\  \mwrs{\mathsf{E}_1}{I}{s} \ \cdots\  \mwrs{\mathsf{E}_n}{I}{s}$
and $\mwrt{\mathsf{E}\theta}{I} = I(\mathsf{f})\ \mwrt{\mathsf{E}_1\theta}{I}\ \cdots\  \mwrt{\mathsf{E}_n\theta}{I}$.
By the induction hypothesis, $\mwrs{\mathsf{E}_1}{I}{s} =
\mwrt{\mathsf{E}_1\theta}{I}, \ldots, \mwrs{\mathsf{E}_n}{I}{s} =
\mwrt{\mathsf{E}_n\theta}{I}$, thus we have
$\mwrs{\mathsf{E}}{I}{s} = \mwrt{\mathsf{E}\theta}{I}$.
Now consider the case that $\mathsf{E}= \mathsf{E}_1\,\mathsf{E}_2$.
We have $\mwrs{\mathsf{E}}{I}{s} = \mwrs{\mathsf{E}_1}{I}{s}\mwrs{\mathsf{E}_2}{I}{s}$
and $\mwrt{\mathsf{E}\theta}{I} = \mwrt{\mathsf{E}_1\theta}{I}\mwrt{\mathsf{E}_2\theta}{I}$.
Again, applying the induction hypothesis, we conclude that
$\mwrs{\mathsf{E}}{I}{s} = \mwrt{\mathsf{E}\theta}{I}$.
Finally, if $\mathsf{E} = (\mathsf{E}_1\approx \mathsf{E}_2)$
we have that $\mwrs{\mathsf{E}}{I}{s}=true$ iff $\mwrs{\mathsf{E}_1}{I}{s}=\mwrs{\mathsf{E}_2}{I}{s}$,
which, by the induction hypothesis, holds iff
$\mwrt{\mathsf{E}_1\theta}{I}=\mwrt{\mathsf{E}_2\theta}{I}$.
Moreover, we have $\mwrt{\mathsf{E}\theta}{I}=true$ iff
$\mwrt{\mathsf{E}_1\theta}{I}=\mwrt{\mathsf{E}_2\theta}{I}$,
therefore we conclude that $\mwrs{\mathsf{E}}{I}{s}=true$ iff
$\mwrt{\mathsf{E}\theta}{I}=true$.
\end{proof}

{Given a Herbrand interpretation $I$ of a definitional program, it is straightforward to
devise a corresponding interpretation of the ground instantiation of the program,
by restricting $I$ to only assigning truth values to ground atoms. As expected,
such a restriction of a model of the program produces a model of its ground
instantiation. This idea is formalized in the following definition and theorem.}
\begin{definition}
Let $\mathsf{P}$ be a program, $I$ be a Herbrand interpretation of $\mathsf{P}$ and $\mathsf{Gr(P)}$
be the ground instantiation of $\mathsf{P}$. We define the \emph{ground restriction}
of $I$, which we denote by $I|_{\mathsf{Gr(P)}}$, to be an interpretation of
$\mathsf{Gr(P)}$, such that, for every ground atom
$\mathsf{A}$, $I|_{\mathsf{Gr(P)}}(\mathsf{A}) = \mwrt{\mathsf{A}}{I}$.
\end{definition}

\begin{theorem}
\label{th_ground_model}
Let $\mathsf{P}$ be a program and $\mathsf{Gr(P)}$ be its ground instantiation.
Also let $M$ be a Herbrand model of $\mathsf{P}$ and $M|_{\mathsf{Gr(P)}}$ be the ground restriction of $M$. Then $M|_\mathsf{Gr(P)}$ is a model of $\mathsf{Gr(P)}$.
\end{theorem}
\begin{proof}
By definition, each clause in $\mathsf{Gr(P)}$ is of the form
$\mathsf{p} \, \mathsf{E}_1 \, \cdots \, \mathsf{E}_n \leftarrow \mathsf{B}_1\theta, \ldots, \mathsf{B}_k\theta$, i.e. the ground instantiation of a clause $\mathsf{p \, V_1 \, \cdots \, V_n} \leftarrow \mathsf{B}_1, \ldots, \mathsf{B}_k$ in $\mathsf{P}$ with respect to a ground substitution $\theta$, such that $dom(\theta)$ includes $V_1,\ldots,V_n$ and all other (individual) variables appearing in the body of the clause and $\theta(\mathsf{V}_i) = \mathsf{E}_i$, for all $i \in \{1, \ldots, n\}$.
Let $s$ be a Herbrand state such that $s(\mathsf{V}) =\mwrt{\theta(\mathsf{V})}{M}$,
for all $V \in dom(\theta)$. By the Substitution Lemma (Lemma \ref{substitution_lemma})
and the definition of $M|_\mathsf{Gr(P)}$,
$\mwrs{\mathsf{p}\, \mathsf{V}_1\, \cdots\, \mathsf{V}_n}{M}{s} =
 \mwrt{\mathsf{p}\, \mathsf{E}_1\, \cdots\, \mathsf{E}_n}{M} =
 M|_\mathsf{Gr(P)}(\mathsf{p}\, \mathsf{E}_1\, \cdots\, \mathsf{E}_n)$.
Similarly, for each atom $\mathsf{B}_i$ in the body of the clause,
we have $\mwrs{\mathsf{B}_i}{M}{s} = \mwrt{\mathsf{B}_i\theta}{M}
= M|_\mathsf{Gr(P)}(\mathsf{B}_i\theta), 1\leq i\leq k$.
Consequently, if $M|_\mathsf{Gr(P)}(\mathsf{B}_i\theta)=true$ for all $i\in\{1,\ldots,k\}$, we also have that $\mwrs{\mathsf{B}_i}{M}{s}=true, 1\leq i\leq k$.
As $M$ is a model of $\mathsf{P}$, this implies that
$\mwrs{\mathsf{p}\, \mathsf{V}_1\, \cdots\, \mathsf{V}_n}{M}{s}=
M|_\mathsf{Gr(P)}(\mathsf{p}\, \mathsf{E}_1\, \cdots\, \mathsf{E}_n)=true$ and
therefore $M|_\mathsf{Gr(P)}$ is a model of $\mathsf{Gr(P)}$.
\end{proof}

{The above theorem is of course useful in connecting the
$\aleq[\mathcal{I}_\mathsf{P}]$-minimum Herbrand model of a program to its ground instantiation. However,
in order to prove the equivalence of the two semantics under consideration,
we will also need to go in the opposite direction and connect the $\leq$-minimum model of the ground program to the higher-order program. To this end we
introduce the previously mentioned \emph{semantic extensions} of a ground
expression.}
\begin{definition}
Let $\mathsf{P}$ be a program and $\bezem$ be the $\leq$-minimum model of $\mathsf{Gr(P)}$.
Let $\mathsf{E}$ be a ground expression of argument type $\rho$ and $d$ be an
element of $\mo{\rho}$. We will say that $d$ is a semantic extension of $E$ and
write $d \ee[\rho] \mathsf{E}$ if
\begin{itemize}
\item $\rho = \iota$ and $d = \mathsf{E}$;
\item $\rho = o$ and $d = \bezem(\mathsf{E})$;
\item $\rho = \rho' \rightarrow \pi$ and for all $d' \in \mo{\rho'}$ and $\mathsf{E}' \in U_{P,\rho'}$,
      such that $d' \ee[\rho'] \mathsf{E}'$, it holds that $d \, d' \ee[\pi] \mathsf{E} \, \mathsf{E}'$.
\end{itemize}
\end{definition}

{Compared to that of the ground restriction presented earlier, the notion of
extending a syntactic object to the realm of semantic elements, is more complicated.
In fact, even the existence of a semantic extension is not immediately obvious.
The next lemma guarantees that not only can such an extension  be constructed
for any expression of the language, but it also has an interesting property of
mirroring the ordering of semantic objects with respect to $\aleq[\tau]$ in a
corresponding ordering of the expressions with respect to $\preceq_\tau$.}
\begin{lemma} \label{lm_order_preservation}
Let $\mathsf{P}$ be a program, $\mathsf{Gr(P)}$ be its ground instantiation and
$\bezem$ be the $\leq$-minimum model of $\mathsf{Gr(P)}$. For every
argument type $\rho$ and every ground term $\mathsf{E} \in U_{\mathsf{P},\rho}$
\begin{enumerate}
\item There exists $e \in \mo{\rho}$ such that $e \ee[\rho] \mathsf{E}$.
\item For all $e, e' \in \mo{\rho}$ and all $\mathsf{E}' \in U_{\mathsf{P},\rho}$, if $e \ee[\rho] \mathsf{E}$, $e' \ee[\rho] \mathsf{E}'$ and $e \aleq[\rho] e'$, then $\mathsf{E} \preceq_\rho \mathsf{E}'$.
\end{enumerate}
\end{lemma}
\begin{proof}
We prove both statements simultaneously, performing an induction on the structure
of $\rho$. Specifically, the first statement is proven by showing that in each
case we can construct a function $e$ of type $\rho$, which is monotonic with
respect to $\aleq[\rho]$ and satisfies $e \ee[\rho] \mathsf{E}$.

In the basis case, the construction of $e$ for types $\iota$ and $o$ is trivial.
Also, if $\rho = \iota$, then both $\ee[\rho]$ and $\aleq[\rho]$ reduce to equality,
so we have $\mathsf{E} = \mathsf{E}'$, which in this case is equivalent to
$\mathsf{E} \preceq_\rho \mathsf{E}'$. On the other hand, for
$\rho = o$, $\ee[\rho]$ identifies with equality, while $\aleq[\rho]$ and
$\preceq_{\rho}$ identify with $\leq$, so we have that
 $\bezem(\mathsf{E}) = e \leq e' = \bezem(\mathsf{E}')$ implies
 $\mathsf{E} \preceq_\rho \mathsf{E}'$.

For a more complex type $\rho = \rho_1 \rightarrow \cdots \rightarrow \rho_n \rightarrow o$, $n > 0$,
we can easily construct $e$, as follows:
\begin{align*}
e \,e_1\,\cdots\,e_n =
  \begin{cases}
      \mathit{true},  & \mbox{if there exist } d_1, \ldots, d_n \mbox{ and ground terms } \mathsf{D}_1,\ldots,\mathsf{D}_n  \mbox{ such that,}\\
             &  \mbox{for all } i, d_i  \aleq[\rho_i] e_i, d_i \ee[\rho_i] \mathsf{D}_i \mbox{ and } \bezem(\mathsf{E}\, \mathsf{D}_1 \, \cdots \, \mathsf{D}_n)=\mathit{true} \\
      \mathit{false}, & \mbox{otherwise.}
  \end{cases}
\end{align*}
To see that $e$ is monotonic, consider $e_1, \ldots, e_n, e'_1, \ldots, e'_n$,
such that $e_1 \aleq[\rho_1] e'_1, \ldots, e_n \aleq[\rho_n] e'_n$ and observe
that $e \,e_1\,\cdots\,e_n = true$ implies $e \,e'_1\,\cdots\,e'_n=true$, due to
the transitivity of $\aleq[\rho_i]$.
We will now show that $e \ee[\rho] \mathsf{E}$, i.e. for all $e_1, \ldots, e_n$ and
$\mathsf{E}_1, \ldots, \mathsf{E}_n$ such that $e_1 \ee[\rho_1] \mathsf{E}_1, \ldots, e_n \ee[\rho_n] \mathsf{E}_n$, it
holds $e \,e_1\,\cdots\,e_n = \bezem(\mathsf{E} \, \mathsf{E}_1 \, \cdots \, \mathsf{E}_n)$. This is trivial
if $\bezem(\mathsf{E} \; \mathsf{E}_1 \, \cdots \, \mathsf{E}_n) = true$, since $e_i \aleq[\rho_i] e_i$. Let
us now examine the case that $\bezem(\mathsf{E} \; \mathsf{E}_1 \, \cdots \, \mathsf{E}_n) = \mathit{false}$. For the
sake of contradiction, assume $e \,e_1\,\cdots\,e_n= true$. Then,
by the construction of $e$, there must exist $d_1, \ldots, d_n$ and $\mathsf{D}_1,\ldots,\mathsf{D}_n$
such that, for all $i$, $d_i  \aleq[\rho_i] e_i$, $d_i \ee[\rho_i] \mathsf{D}_i$ and $\bezem(\mathsf{E}\, \mathsf{D}_1 \, \cdots \, \mathsf{D}_n)=true$. By the induction hypothesis, we have that $\mathsf{D}_i \preceq_{\rho_i} \mathsf{E}_i$, for all $i \in \{1, \ldots, n\}$. This, by the $\preceq$-Monotonicity Property of $\bezem$ (Theorem \ref{th-preceq-monotonicity}), yields that $\bezem(\mathsf{E}\, \mathsf{D}_1 \, \cdots \, \mathsf{D}_n)=true \leq \bezem(\mathsf{E} \; \mathsf{E}_1 \, \cdots \, \mathsf{E}_n) = \mathit{false}$, which is obviously a contradiction. Therefore it has to be that $e \,e_1\,\cdots\,e_n= \mathit{false}$.

Finally, in order to prove the second statement and conclude the induction step, we need to show that for all terms $\mathsf{D}_1 \in U_{\mathsf{P},\rho_1}, \ldots, \mathsf{D}_n \in U_{\mathsf{P},\rho_n}$,
it holds $\mathsf{E} \; \mathsf{D}_1 \, \cdots \, \mathsf{D}_n \preceq_o \mathsf{E}' \; \mathsf{D}_1 \, \cdots \, \mathsf{D}_n$. By the induction hypothesis, there exist $d_1, \ldots, d_n$, such that $d_1 \ee[\rho_1] \mathsf{D}_1, \ldots, d_n \ee[\rho_n] \mathsf{D}_n$.
Because $e \ee[\rho] \mathsf{E}$ and $\mathsf{E} \; \mathsf{D}_1 \, \cdots \, \mathsf{D}_n$ is of type $o$, we
have $e \; d_1 \, \cdots \, d_n = \bezem(\mathsf{E} \; \mathsf{D}_1 \, \cdots \, \mathsf{D}_n)$ by definition.
Similarly, we also have $e' \; d_1 \, \cdots \, d_n = \bezem(\mathsf{E}' \; \mathsf{D}_1 \, \cdots \, \mathsf{D}_n)$. Moreover,
by $e \aleq[\rho] e'$ we have that
$e \; d_1 \, \cdots \, d_n \aleq[o] e' \; d_1 \, \cdots \, d_n$. This yields the desired result, since $\aleq[o]$ identifies with $\preceq_o$.
\end{proof}

{The following variation of the Substitution Lemma states that if the building elements of an expression are assigned meanings that are semantic extensions of their syntactic counterparts, then the meaning of the expression is itself a semantic extension of the expression.}
\begin{lemma} \label{lm_state_substitution_extension}
Let $\mathsf{P}$ be a program, $\mathsf{Gr(P)}$ be its ground instantiation and $I$ be a Herbrand
interpretation of $\mathsf{P}$. Also, let $\mathsf{E}$ be an expression
of some argument type $\rho$ and
let $s$ be a Herbrand state and $\theta$ be a ground substitution, both with domain $vars(\mathsf{E})$.
If, for all predicates $\mathsf{p}$ of type $\pi$ appearing in $\mathsf{E}$, $\mwrt{\mathsf{p}}{I} \ee[\pi] \mathsf{p}$
and, for all variables $\mathsf{V}$ of type $\rho'$ in $vars(\mathsf{E})$, $s(\mathsf{V}) \ee[\rho'] \theta(\mathsf{V})$, then
$\mwrs{\mathsf{E}}{I}{s} \ee[\rho] \mathsf{E}\theta$.
\end{lemma}
\begin{proof}
The proof is by induction on the structure of $\mathsf{E}$.
The basis cases $\mathsf{E} = \mathsf{p}$ and $\mathsf{E} = \mathsf{V}$
hold by assumption and $\mathsf{E} = c:\iota$ is trivial.
For the first case of the induction step, let $\mathsf{E} =
(\mathsf{f}\ \mathsf{E}_1 \ \cdots \  \mathsf{E}_n)$,
where $\mathsf{E}_1,\ldots, \mathsf{E}_n$ are of type
$\iota$. By the induction hypothesis, we have that
$\mwrs{\mathsf{E}_1}{I}{s} \ee[\iota] \mathsf{E}_1\theta, \ldots, \mwrs{\mathsf{E}_n}{I}{s} \ee[\iota] \mathsf{E}_n\theta$.
As $\ee[\iota]$ is defined as equality, we have that $\mwrs{\mathsf{E}}{I}{s} = I(\mathsf{f})\ \mwrs{\mathsf{E}_1}{I}{s}\ \cdots \ \mwrs{\mathsf{E}_n}{I}{s} = \mathsf{f}\ \mathsf{E}_1\theta\ \cdots \ \mathsf{E}_n\theta = \mathsf{E}\theta$ and therefore $\mwrs{\mathsf{E}}{I}{s} \ee[\iota] \mathsf{E}\theta$.
For the second case, let $\mathsf{E} = \mathsf{E}_1 \, \mathsf{E}_2$,
where $\mathsf{E}_1$ is of type $\rho_1 = \rho_2\rightarrow \pi$ and $\mathsf{E}_2$ is of type $\rho_2$; then, $\mwrs{\mathsf{E}}{I}{s} = \mwrs{\mathsf{E}_1}{I}{s}\, \mwrs{\mathsf{E}_2}{I}{s}$.
By the induction hypothesis, $\mwrs{\mathsf{E}_1}{I}{s} \ee[\rho_2\rightarrow \pi] \mathsf{E}_1\theta$
and $\mwrs{\mathsf{E}_2}{I}{s} \ee[\rho_2] \mathsf{E}_2\theta$,
thus, by definition, $\mwrs{\mathsf{E}}{I}{s}=\mwrs{\mathsf{E}_1}{I}{s}\;\mwrs{\mathsf{E}_2}{I}{s} \ee[\pi] \mathsf{E}_1\theta \; \mathsf{E}_2\theta = (\mathsf{E}_1 \, \mathsf{E}_2)\theta=\mathsf{E}\theta$.
Finally, we have the case that $\mathsf{E} = (\mathsf{E}_1 \approx \mathsf{E}_2)$, where $\mathsf{E}_1$ and $\mathsf{E}_2$ are both of type $\iota$. The induction hypothesis yields $\mwrs{\mathsf{E}_1}{I}{s} \ee[\iota] \mathsf{E}_1\theta$ and $\mwrs{\mathsf{E}_2}{I}{s} \ee[\iota] \mathsf{E}_2\theta$ or, since $\ee[\iota]$ is defined as equality, $\mwrs{\mathsf{E}_1}{I}{s} = \mathsf{E}_1\theta$ and $\mwrs{\mathsf{E}_2}{I}{s} = \mathsf{E}_2\theta$. Then  $\mwrs{\mathsf{E}_1}{I}{s} = \mwrs{\mathsf{E}_2}{I}{s}$ iff  $\mathsf{E}_1\theta = \mathsf{E}_2\theta$ and, equivalently, $\mwrs{\mathsf{E}}{I}{s} = true$ iff $\mathsf{E}\theta = true$, which implies $\mwrs{\mathsf{E}}{I}{s} \ee[o] \mathsf{E}\theta$.
\end{proof}

{We are now ready to present the main result of this paper. The theorem
establishes the equivalence of Wadge's semantics and Bezem's semantics, in
stating that their respective minimum models assign the same meaning to all ground atoms.}
\begin{theorem}
Let $\mathsf{P}$ be a program and let $\mathsf{Gr(P)}$ be its ground instantiation.
Let $M_\mathsf{P}$ be the $\aleq[\mathcal{I}_\mathsf{P}]$-minimum Herbrand model of $\mathsf{P}$ and let $\bezem$ be the
$\leq$-minimum model of $\mathsf{Gr(P)}$. Then, for every $\mathsf{A} \in U_{\mathsf{P},o}$
it holds $\mwrt{\mathsf{A}}{M_\mathsf{P}} = \bezem(\mathsf{A})$.
\end{theorem}
\begin{proof}

We will construct an interpretation $N$ for $\mathsf{P}$ and prove
some key properties for this interpretation. Then we will utilize these properties
to prove the desired result. The definition of $N$ is as follows:
\begin{align*}
& \mbox{For every } \mathsf{p} : \rho_1 \rightarrow \cdots \rightarrow \rho_n \rightarrow o \mbox{ and all } d_1 \in \mo{\rho_1}, \ldots, d_n \in \mo{\rho_n}\\
& N(\mathsf{p})\,d_1\,\cdots\,d_n =
   \begin{cases}
    \mathit{false},  & \mbox{if there exist } e_1, \ldots, e_n \mbox{ and ground terms } \mathsf{E}_1,\ldots,\mathsf{E}_n \mbox{ such that,} \\
            & \mbox{for all } i, d_i  \aleq[\rho_i] e_i, e_i \ee[\rho_i] \mathsf{E}_i \mbox{ and } \bezem(\mathsf{p}\, \mathsf{E}_1 \, \cdots \, \mathsf{E}_n)=\mathit{false}\\
    true,   & \mbox{otherwise}
   \end{cases}
\end{align*}

Observe that $N$ is a valid Herbrand interpretation of $\mathsf{P}$, in the sense
that it assigns elements in $\mo{\pi}$ (i.e. functions that are monotonic with
respect to $\aleq[\pi]$) to every predicate of type $\pi$ in $\mathsf{P}$. Indeed,
if it was not so, then for some predicate $\mathsf{p} : \pi= \rho_1 \rightarrow \cdots \rightarrow \rho_n \rightarrow o$,
there would exist tuples $(d_1, \ldots, d_n)$ and $(d'_1, \ldots, d'_n)$ with
$d_1 \aleq[\rho_1] d'_1, \ldots, d_n \aleq[\rho_n] d'_n$, such that
$N(\mathsf{p})\,d_1\,\cdots\,d_n = true$ and $N(\mathsf{p})\,d'_1\,\cdots\,d'_n = \mathit{false}$.
By definition, the fact that $N(\mathsf{p})\,d'_1\,\cdots\,d'_n$ is assigned the value
$\mathit{false}$, would imply that there exist $e_1, \ldots, e_n$ and $\mathsf{E}_1,\ldots,\mathsf{E}_n$
as in the above definition, such that $\bezem(\mathsf{p}\, \mathsf{E}_1 \, \cdots \, \mathsf{E}_n)=\mathit{false}$ and
$d'_1 \aleq[\rho_1] e_1, \ldots, d'_n \aleq[\rho_n] e_n$. Being that $\aleq[\rho_i]$
are transitive relations, the latter yields that
$d_1 \aleq[\rho_1] e_1, \ldots, d_n \aleq[\rho_n] e_n$. Therefore, by definition,
$N(\mathsf{p})\,d_1\,\cdots\,d_n$ should also evaluate to $\mathit{false}$, which constitutes a
contradiction and thus confirms that the meaning of $p$ is monotonic with respect to $\aleq[\pi]$.

It is also straightforward to see that $N(\mathsf{p}) \ee[\pi] \mathsf{p}$, i.e.
for all $d_1, \ldots, d_n$ and all ground terms $\mathsf{D}_1, \ldots, \mathsf{D}_n$
such that $d_1 \ee[\rho_1] \mathsf{D}_1, \ldots, d_n \ee[\rho_n] \mathsf{D}_n$, we have
$N(\mathsf{p}) \,d_1\,\cdots\,d_n = \bezem(\mathsf{p} \; \mathsf{D}_1 \, \cdots \, \mathsf{D}_n)$. Because
$d_i \aleq[\rho_i] d_i$, this holds trivially if $\bezem(\mathsf{p} \; \mathsf{D}_1 \, \cdots \, \mathsf{D}_n) = \mathit{false}$.
Now let $\bezem(\mathsf{p} \; \mathsf{D}_1 \, \cdots \, \mathsf{D}_n) = true$ and assume, for the sake of
contradiction, that $N(\mathsf{p}) \,d_1\,\cdots\,d_n = \mathit{false}$. Then, by the definition
of $N$, there must exist $e_1, \ldots, e_n$ and $\mathsf{E}_1,\ldots,\mathsf{E}_n$ such that, for
all $i$, $d_i  \aleq[\rho_i] e_i$, $e_i \ee[\rho_i] \mathsf{E}_i$ and $\bezem(\mathsf{p}\, \mathsf{E}_1 \, \cdots \, \mathsf{E}_n)=\mathit{false}$.
Thus, by the second part of Lemma \ref{lm_order_preservation}, for all $i$, $\mathsf{D}_i \preceq_{\rho_i} \mathsf{E}_i$ and,
by the $\preceq$-Monotonicity Property of $\bezem$, $\bezem(\mathsf{p}\, \mathsf{D}_1 \, \cdots \, \mathsf{D}_n) \leq \bezem(\mathsf{p} \; \mathsf{E}_1 \, \cdots \, \mathsf{E}_n)$,
which is obviously a contradiction. Thus we conclude that $N(\mathsf{p}) \,d_1\,\cdots\,d_n=true$.

Next we prove that $N$ is a model of $\mathsf{P}$. Let $\mathsf{p} \, \mathsf{V}_1 \, \cdots \, \mathsf{V}_n \leftarrow \mathsf{B}_1,\ldots,\mathsf{B}_k$ be a clause in $\mathsf{P}$ and let $\{\mathsf{V}_1, \ldots, \mathsf{V}_n, \mathsf{X}_1, \ldots, \mathsf{X}_m\}$, with $\mathsf{V}_i:\rho_i$, for all $i\in \{1, \ldots,n\}$, and $\mathsf{X}_i:\iota$, for all $i\in \{1, \ldots,m\}$, be the set of variables appearing in the clause. Then, it suffices to show that, for any tuple $(d_1, \ldots, d_n)$ of arguments and any Herbrand state $s$ such that $s(\mathsf{V}_i) = d_i$ for all $i\in\{1,\ldots, n\}$, $N(\mathsf{p})\,d_1\,\cdots\,d_n = \mathit{false}$ implies that, for at least one $j \in \{1,\ldots,k\}$, $\mwrs{\mathsf{B}_j}{N}{s} = \mathit{false}$. Again, by the definition of $N$, we see that if $N(\mathsf{p})\,d_1\,\cdots\,d_n = \mathit{false}$, then there exist $e_1, \ldots, e_n$ and ground terms $\mathsf{E}_1,\ldots,\mathsf{E}_n$ such that $\bezem(\mathsf{p}\, \mathsf{E}_1 \, \cdots \, \mathsf{E}_n)=\mathit{false}$, $d_1 \aleq[\rho_1] e_1, \ldots, d_n \aleq[\rho_n] e_n$ and $e_1 \ee[\rho_1] \mathsf{E}_1, \ldots, d_n \ee[\rho_n] \mathsf{E}_n$. Let $\theta$ be a ground substitution such that $\theta(\mathsf{V}_i) = \mathsf{E}_i$ for all $i\in\{1,\ldots,n\}$ and, for all $i\in\{1,\ldots,m\}$, $\theta(\mathsf{X}_i) = s(\mathsf{X}_i)$; then there exists a ground instantiation $\mathsf{p} \, \mathsf{E}_1 \, \cdots \, \mathsf{E}_n \leftarrow \mathsf{B}_1\theta,\ldots,\mathsf{B}_k\theta$ of the above clause in $\mathsf{Gr(P)}$. As $\bezem$ is a model of the ground program, $\bezem(\mathsf{p}\, \mathsf{E}_1 \, \cdots \, \mathsf{E}_n)=\mathit{false}$ implies that there exists at least one $j \in \{1,\ldots,k\}$ such that $\bezem(\mathsf{B}_j\theta)=\mathit{false}$. We are going to show that the latter implies that $\mwrs{\mathsf{B}_j}{N}{s} = \mathit{false}$, which proves that $N$ is a model of $\mathsf{P}$. Indeed, let $s'$ be a Herbrand state such that $s'(\mathsf{V}_i) = e_i \ee[\rho_i] \theta(\mathsf{V}_i) = \mathsf{E}_i$ for all $i\in\{1,\ldots, n\}$ and $s'(\mathsf{X}_i)=\theta(\mathsf{X}_i)=s(\mathsf{X}_i)$ for all $i\in\{1,\ldots,m\}$. As we have shown earlier, $N(\mathsf{p}') \ee[\pi'] \mathsf{p}'$ for any predicate $\mathsf{p}':\pi'$, thus by Lemma \ref{lm_state_substitution_extension} we get $\mwrs{\mathsf{B}_j}{N}{s'} \ee[o] \mathsf{B}_j\theta$. Since $\mathsf{B}_j$ is of type $o$, the latter reduces to $\mwrs{\mathsf{B}_j}{N}{s'} = \bezem(\mathsf{B}_j\theta) = \mathit{false}$. Also, because $d_i \aleq[\rho_i] e_i$, i.e. $s \aleq[\mathcal{S}_\mathsf{P}] s'$, by the second part of Lemma \ref{interp-state-monotonicity} we get $\mwrs{\mathsf{B}_j}{N}{s} \aleq[o] \mwrs{\mathsf{B}_j}{N}{s'}$, which makes $\mwrs{\mathsf{B}_j}{N}{s} = \mathit{false}$.

Now we can proceed to prove that, for all $\mathsf{A} \in U_{\mathsf{P},o}$, $\mwrt{\mathsf{A}}{M_\mathsf{P}} = \bezem(\mathsf{A})$.
Let $\mathsf{A}$ be of the form $\mathsf{p} \; \mathsf{E}_1 \, \cdots \, \mathsf{E}_n$, where
$\mathsf{p} : \rho_1 \rightarrow \cdots \rightarrow \rho_n \rightarrow o \in \mathsf{P}$ and let
$d_1 = \mwrt{\mathsf{E}_1}{M_\mathsf{P}}, \ldots, d_n = \mwrt{\mathsf{E}_n}{M_\mathsf{P}}$. As we have shown,
$N$ is a Herbrand model of $\mathsf{P}$, while $M_\mathsf{P}$ is the minimum, with respect to $\aleq[\mathcal{I}_\mathsf{P}]$, of all Herbrand models of $\mathsf{P}$,
therefore we have that $M_\mathsf{P} \aleq[\mathcal{I}_\mathsf{P}] N$. By definition, this gives us that
$M_\mathsf{P}(\mathsf{p})\; d_1 \, \cdots \, d_n \aleq[o] N(\mathsf{p}) \; d_1 \, \cdots \, d_n \; (1)$ and,
by the first part of Lemma \ref{interp-state-monotonicity}, that
$d_1 \aleq[\rho_1] \mwrt{\mathsf{E}_1}{N}, \ldots, d_n \aleq[\rho_n] \mwrt{\mathsf{E}_n}{N} \;(2)$.
Moreover, for all predicates $\mathsf{p}':\pi'$ in $\mathsf{P}$, we have $N(\mathsf{p}') \ee[\pi'] \mathsf{p}'$ and thus,
by Lemma \ref{lm_state_substitution_extension}, taking $s$ and $\theta$ to be empty,
we get $\mwrt{\mathsf{E}_i}{N} \ee[\rho_i] \mathsf{E}_i, 1\leq i\leq n$. In conjunction with (2), the latter suggests
that if $\bezem(\mathsf{p} \; \mathsf{E}_1 \; \cdots \; \mathsf{E}_n) = \mathit{false}$ then $N(\mathsf{p}) \; d_1 \; \cdots \; d_n = \mathit{false}$,
or, in other words, that $ N(\mathsf{p}) \; d_1 \, \cdots \, d_n \leq \bezem(\mathsf{p} \; \mathsf{E}_1 \; \cdots \; \mathsf{E}_n)$.
Because of (1), this makes it that
$ M_\mathsf{P}(\mathsf{p}) \; d_1 \, \cdots \, d_n \leq \bezem(\mathsf{p} \; \mathsf{E}_1 \; \cdots \; \mathsf{E}_n)\;(3)$. On the
other hand, by Theorem \ref{th_ground_model}, $M_\mathsf{P}|_\mathsf{Gr(P)}$ is a model of $\mathsf{Gr(P)}$ and therefore $\bezem(\mathsf{p} \; \mathsf{E}_1 \; \cdots \; \mathsf{E}_n) \leq M_\mathsf{P}|_\mathsf{Gr(P)}(\mathsf{p} \; \mathsf{E}_1 \; \cdots \; \mathsf{E}_n)$, since $\bezem$ is the minimum model of $\mathsf{Gr(P)}$. By the definition of $\mathsf{Gr(P)}$ and the meaning of application, the latter becomes
$\bezem(\mathsf{p} \; \mathsf{E}_1 \; \cdots \; \mathsf{E}_n) \leq M_\mathsf{P}|_\mathsf{Gr(P)}(\mathsf{p} \; \mathsf{E}_1 \; \cdots \; \mathsf{E}_n)=M_\mathsf{P}(\mathsf{p} \; \mathsf{E}_1 \; \cdots \; \mathsf{E}_n) = M_\mathsf{P}(\mathsf{p})\; \mwrt{\mathsf{E}_1}{M_\mathsf{P}}\; \cdots \; \mwrt{\mathsf{E}_n}{M_\mathsf{P}} = M_\mathsf{P}(\mathsf{p})\; d_1 \, \cdots \, d_n$. The last relation and (3) can only be true simultaneously, if all the above relations hold as equalities, in particular if $\bezem(\mathsf{p} \; \mathsf{E}_1 \, \cdots \, \mathsf{E}_n) = \mwrt{\mathsf{p} \; \mathsf{E}_1 \, \cdots \, \mathsf{E}_n}{M_\mathsf{P}}$.
\end{proof}

\section{Discussion}
\label{sec:concl}
\label{section5}
We have considered the two existing extensional approaches to the semantics of higher-order
logic programming, and have demonstrated that they coincide for the class of definitional
programs. It is therefore natural to wonder whether the two semantic approaches continue
to coincide if we extend the class of programs we consider. Unfortunately this is not
the case, as we discuss below.

A seemingly mild extension to our source language would be to allow higher-order predicate
variables that are not formal parameters of a clause, to appear in its body. Such programs
are legitimate under Bezem's semantics (ie., they belong to the hoapata class). Moreover,
a recent extension of Wadge's semantics~\cite{CharalambidisHRW13} also allows such programs.
However, for this extended class of programs the equivalence of the two semantic approaches
no longer holds as the following example illustrates.
\begin{example}
Consider the following extended program:
\[
\begin{array}{l}
\mbox{\tt p(a):-Q(a).}
\end{array}
\]
Following Bezem's semantics, we initially take the ground instantiation of the program, namely:
\[
\begin{array}{l}
\mbox{\tt p(a):-p(a).}
\end{array}
\]
and then compute the least model of the above program which assigns to the atom
{\tt p(a)} the value $\mathit{false}$. On the other hand, under the approach
in~\cite{CharalambidisHRW13}, the atom {\tt p(a)} has the value $\mathit{true}$
in the minimum Herbrand model of the initial program. This is due to the fact that
under the semantics of~\cite{CharalambidisHRW13}, our initial program reads (intuitively
speaking) as follows: ``{\tt p(a)} is true if there exists a relation that is
true of {\tt a}''; actually, there exists one such relation, namely the set $\{{\tt a}\}$.
This discrepancy between the two semantics is due to the fact that Wadge's semantics is
based on {\em sets} and not solely on the syntactic entities that appear in the
program.\qed
\end{example}

Future work includes the extension of Bezem's approach to higher-order logic programs
with negation. An extension of Wadge's approach for such programs has recently
been performed in~\cite{CharalambidisER14}. More generally, the addition of negation
to higher-order logic programming appears to offer an interesting and nontrivial
area of research, which we are currently pursuing.

\bibliographystyle{eptcs}
\bibliography{fics15}

\begin{thebibliography}{10}
\providecommand{\bibitemdeclare}[2]{}
\providecommand{\surnamestart}{}
\providecommand{\surnameend}{}
\providecommand{\urlprefix}{Available at }
\providecommand{\url}[1]{\texttt{#1}}
\providecommand{\href}[2]{\texttt{#2}}
\providecommand{\urlalt}[2]{\href{#1}{#2}}
\providecommand{\doi}[1]{doi:\urlalt{http://dx.doi.org/#1}{#1}}
\providecommand{\bibinfo}[2]{#2}

\bibitemdeclare{inproceedings}{Bezem99}
\bibitem{Bezem99}
\bibinfo{author}{Marc \surnamestart Bezem\surnameend} (\bibinfo{year}{1999}):
  \emph{\bibinfo{title}{Extensionality of Simply Typed Logic Programs}}.
\newblock In \bibinfo{editor}{Danny~De \surnamestart Schreye\surnameend},
  editor: {\sl \bibinfo{booktitle}{Logic Programming: The 1999 International
  Conference, Las Cruces, New Mexico, USA, November 29 - December 4, 1999}},
  \bibinfo{publisher}{{MIT} Press}, pp. \bibinfo{pages}{395--410}.

\bibitemdeclare{inproceedings}{Bezem01}
\bibitem{Bezem01}
\bibinfo{author}{Marc \surnamestart Bezem\surnameend} (\bibinfo{year}{2001}):
  \emph{\bibinfo{title}{An Improved Extensionality Criterion for Higher-Order
  Logic Programs}}.
\newblock In \bibinfo{editor}{Laurent \surnamestart Fribourg\surnameend},
  editor: {\sl \bibinfo{booktitle}{Computer Science Logic, 15th International
  Workshop, {CSL} 2001. 10th Annual Conference of the EACSL, Paris, France,
  September 10-13, 2001, Proceedings}}, {\sl \bibinfo{series}{Lecture Notes in
  Computer Science}} \bibinfo{volume}{2142}, \bibinfo{publisher}{Springer}, pp.
  \bibinfo{pages}{203--216}, \doi{10.1007/3-540-44802-0\_15}.

\bibitemdeclare{inproceedings}{Bezem2002}
\bibitem{Bezem2002}
\bibinfo{author}{Marc \surnamestart Bezem\surnameend} (\bibinfo{year}{2002}):
  \emph{\bibinfo{title}{Hoapata programs are monotonic}}.
\newblock In: {\sl \bibinfo{booktitle}{Proceedings NWPT’02, Institute of
  Cybernetics at TTU, Tallinn}}, pp. \bibinfo{pages}{18--20}.

\bibitemdeclare{article}{CharalambidisER14}
\bibitem{CharalambidisER14}
\bibinfo{author}{Angelos \surnamestart Charalambidis\surnameend},
  \bibinfo{author}{Zolt{\'{a}}n \surnamestart {\'{E}}sik\surnameend} \&
  \bibinfo{author}{Panos \surnamestart Rondogiannis\surnameend}
  (\bibinfo{year}{2014}): \emph{\bibinfo{title}{Minimum Model Semantics for
  Extensional Higher-order Logic Programming with Negation}}.
\newblock {\sl \bibinfo{journal}{{TPLP}}}
  \bibinfo{volume}{14}(\bibinfo{number}{4-5}), pp. \bibinfo{pages}{725--737},
  \doi{10.1017/S1471068414000313}.

\bibitemdeclare{article}{CharalambidisHRW13}
\bibitem{CharalambidisHRW13}
\bibinfo{author}{Angelos \surnamestart Charalambidis\surnameend},
  \bibinfo{author}{Konstantinos \surnamestart Handjopoulos\surnameend},
  \bibinfo{author}{Panagiotis \surnamestart Rondogiannis\surnameend} \&
  \bibinfo{author}{William~W. \surnamestart Wadge\surnameend}
  (\bibinfo{year}{2013}): \emph{\bibinfo{title}{Extensional Higher-Order Logic
  Programming}}.
\newblock {\sl \bibinfo{journal}{{ACM} Trans. Comput. Log.}}
  \bibinfo{volume}{14}(\bibinfo{number}{3}), p.~\bibinfo{pages}{21},
  \doi{10.1145/2499937.2499942}.

\bibitemdeclare{article}{CKW93-187}
\bibitem{CKW93-187}
\bibinfo{author}{Weidong \surnamestart Chen\surnameend},
  \bibinfo{author}{Michael \surnamestart Kifer\surnameend} \&
  \bibinfo{author}{David~Scott \surnamestart Warren\surnameend}
  (\bibinfo{year}{1993}): \emph{\bibinfo{title}{{HILOG:} {A} Foundation for
  Higher-Order Logic Programming}}.
\newblock {\sl \bibinfo{journal}{Journal of Logic Programming}}
  \bibinfo{volume}{15}(\bibinfo{number}{3}), pp. \bibinfo{pages}{187--230},
  \doi{10.1016/0743-1066(93)90039-J}.

\bibitemdeclare{inproceedings}{KRW05}
\bibitem{KRW05}
\bibinfo{author}{Vassilis \surnamestart Kountouriotis\surnameend},
  \bibinfo{author}{Panos \surnamestart Rondogiannis\surnameend} \&
  \bibinfo{author}{William~W. \surnamestart Wadge\surnameend}
  (\bibinfo{year}{2005}): \emph{\bibinfo{title}{Extensional Higher-Order
  Datalog}}.
\newblock In: {\sl \bibinfo{booktitle}{Short Paper Proceeding of the 12th
  International Conference on Logic for Programming, Artificial Intelligence
  and Reasoning (LPAR)}}, pp. \bibinfo{pages}{1--5}.

\bibitemdeclare{book}{lloyd}
\bibitem{lloyd}
\bibinfo{author}{John~W. \surnamestart Lloyd\surnameend}
  (\bibinfo{year}{1987}): \emph{\bibinfo{title}{Foundations of Logic
  Programming}}.
\newblock \bibinfo{publisher}{Springer Verlag},
  \doi{10.1007/978-3-642-83189-8}.

\bibitemdeclare{book}{MN2012}
\bibitem{MN2012}
\bibinfo{author}{Dale \surnamestart Miller\surnameend} \&
  \bibinfo{author}{Gopalan \surnamestart Nadathur\surnameend}
  (\bibinfo{year}{2012}): \emph{\bibinfo{title}{Programming with Higher-Order
  Logic}}, \bibinfo{edition}{1st} edition.
\newblock \bibinfo{publisher}{Cambridge University Press},
  \bibinfo{address}{New York, NY, USA}, \doi{10.1017/CBO9781139021326}.

\bibitemdeclare{inproceedings}{Wa91a}
\bibitem{Wa91a}
\bibinfo{author}{William~W. \surnamestart Wadge\surnameend}
  (\bibinfo{year}{1991}): \emph{\bibinfo{title}{Higher-Order Horn Logic
  Programming}}.
\newblock In \bibinfo{editor}{Vijay~A. \surnamestart Saraswat\surnameend} \&
  \bibinfo{editor}{Kazunori \surnamestart Ueda\surnameend}, editors: {\sl
  \bibinfo{booktitle}{Logic Programming, Proceedings of the 1991 International
  Symposium, San Diego, California, USA, Oct. 28 - Nov 1, 1991}},
  \bibinfo{publisher}{{MIT} Press}, pp. \bibinfo{pages}{289--303}.

\end{thebibliography}

\appendix

\end{document}